\documentclass[oneside]{amsart}
\pdfoutput=1

\usepackage{tikz,tikz-3dplot}
\usetikzlibrary{positioning,arrows.meta,decorations.pathreplacing,decorations.pathmorphing,decorations.markings,intersections}
\tikzset{li/.style={line width=1pt}}

\usepackage{graphicx,amssymb,amsfonts,amsmath,amsthm,newlfont}

\newcommand{\CC}{{\mathbb C}}

\newcommand{\FF}{{\mathbb F}}

\newcommand{\RR}{{\mathbb R}}

\newcommand{\ee}{\mathrm{e}}

\newcommand{\dd}{\mathrm{d}}

\newcommand{\zz}{{\overline{z}}}

\newcommand{\sE}{\mathcal{E}}

\newcommand{\sV}{\mathcal{V}}

\theoremstyle{plain}
\newtheorem{thm}{Theorem}
\newtheorem{lem}[thm]{Lemma}

\newtheorem{cor}[thm]{Corollary}
\newtheorem{prop}[thm]{Proposition}
\newtheorem{quest}[thm]{Question}
\newtheorem{remark}[thm]{Remark}
\newtheorem{defn}[thm]{Definition}

\title{The five-twist identity for Feynman periods}
\date{}
\author{Oliver Schnetz}
\address{Oliver Schnetz\\
II. Institut f\"ur Theoretische Physik\\
Luruper Chaussee 149\\
22761 Hamburg, Germany}
\email{schnetz@mi.uni-erlangen.de}

\begin{document}

\begin{abstract}
We prove a new identity for Feynman periods that acts on five-vertex cuts of completed primitive Feynman graphs.
It is shown that in $\phi^4$ theory this identity is independent from existing identities which are the twist, the Fourier identity and the Fourier split.
\end{abstract}
\maketitle

\section{Introduction}
We consider a renormalizable quantum field theory (QFT) with dimensionless couplings. Such a QFT has a conformal symmetry in integer dimensions that is typically
broken the the existence of divergences, so that the QFT has to be renormalized (often by generalizing to 'non-integer dimensions'.
In particular, Feynman graphs that have the external structure of a vertex relate to Feynman integrals that have an overall logarithmic divergence
(if they are not forests). The simplest case are graphs with no additional (sub-)divergences.
Graphs of this type is called primitive in the set of Feynman graphs of the corresponding QFT.
With Definition \ref{compprim} and Proposition \ref{prop:conv1} we formulate primitivity in a purely graph theoretical setting without reference to Feynman integrals.

Residues of primitive graphs give rise to renormalization scheme independent contributions to the $\beta$-function of the vertex interaction;
the Feynman period of the graph \cite{BK,Census}.


Concretely, we fix a primitive graph $G$ and deletes the external legs (which carry the vertex structure) because they are insignificant for the calculation of the Feynman period.
In this article, we restrict ourselves to scalar QFTs (with spin zero bosons and no fermions).
Feynman periods in four-dimensional Yukawa-$\phi^4$ theory which has a spin zero boson and a spin $1/2$ fermion have been calculated in \cite{Sgft}.
Although in this article we are mostly interested in four-dimensional $\phi^4$ theory (Sections \ref{sect3a} and \ref{sect4}) it is natural to cosider integer dimensions
\begin{equation}
D=2\lambda+2\geq4
\end{equation}
and graphs with vertices of any degree and edges $e$ of any weights $\nu_e\in\RR$. In a weighted graph, the weighted degree of a vertex $v$ is the sum of
the weights of edges adjacent to $v$,
\begin{equation}\label{Nv}
N_v=\sum_{e\sim v}\nu_e.
\end{equation}
If all edges have weight $1$, the weighted degree of $v$ equals the degree of $v$, i.e. the number of edges adjacent to $v$.

Note that in this setup, Feynman graphs with tensor structure can be reduced to scalar integrals, so that the results cover a rather broad range of
applications that ultimately also includes Gauge theories.

In position space, the Feynman propagator that corresponds to the edge $e=xy$ with vertices $x,y\in\RR^D$ and weight $\nu_e$ is
\begin{equation}\label{eqpe}
p_e(x,y)=\frac1{||x-y||^{2\lambda\nu_e}}.
\end{equation}
Because the signature of the norm $||\cdot||$ has no effect on the Feynman period, we restrict ourselves to the Euclidean metric
\begin{equation}
||x||^2=x_1^2+\ldots+x_D^2.
\end{equation}
In a scalar theory, the vertex has no structure and we integrate over all vertices except for one vertex $0$ that is the origin of the coordinate system
(to break translational symmetry) and one vertex $1$ that we associate to any unit vector $z_1\in\RR^D$ (to break the scale symmetry). The rotational symmetry ensures
that the Feynman period does not depend on the orientation of $z_1$. The Feynman period $P_G$ of the primitive graph $G$ is the (convergent) integral
\begin{equation}\label{PG}
P_G=\Big(\prod_{v\neq0,1}\int_{\RR^D}\frac{\dd^Dx_v}{\pi^{D/2}}\Big)\prod_{e\in\sE_G}p_e(x)\in\RR_+,
\end{equation}
where the integral is over all `internal' vertices $\neq0,1$ and the integrand is the product over the propagators of the edges in $G$.
The Feynman period $P_G$ does not depend on the choice of the vertices $0$ and $1$.
In physics, this statement is known as cut and glue identity for the graph (the massless-$p$ integral) with external legs at 0 and 1 and the edge 01 deleted (if present).

The Feynman period $P_G$ can be expressed in terms of an integral over the graph polynomial by using Schwinger (Feynman) parameters.
This parametric expression for $P_G$ is most frequently used the literature, so that we indicate the connection although we will not need the results in the following.

The graph polynomial is defined as \cite{KIR}
\begin{equation}\label{PsiG}
\Psi_G(\alpha)=\sum_{T\subseteq G}\prod_{e\notin T}\alpha_e,
\end{equation}
where the sum is over all spanning trees in $G$. If all weighs are positive, $\nu_e>0$, we obtain a representation of the Feynman period in terms of a
projective integral (Corollary 26 in \cite{gfe} and the remark thereafter),
\begin{equation}\label{para}
P_G=\frac{\Gamma(\lambda+1)}{\prod_{e\in\sE_G}\Gamma(\lambda\nu_e)}\int_{\alpha_e>0}\Omega\frac{\prod_{e\in\sE_G}\alpha_e^{\lambda(1-\nu_e)}}{\Psi_G(\alpha)^{\lambda+1}}.
\end{equation}
Here, $\Gamma(x)=\int_0^\infty t^{x-1}\ee^{-t}\dd t$ is the gamma function and
\begin{equation}
\Omega=\sum_{e=1}^{|\sE_G|}(-1)^{e-1}\alpha_e\dd \alpha_1\wedge\ldots\wedge\widehat{\dd \alpha_e}\wedge\ldots\wedge\dd \alpha_{|\sE_G|}
\end{equation}
is the projective volume form. The integration is over the projective positive coordinate simplex. In practice, one works in an affine chart by setting one of the
variables to one, e.g.\ $\alpha_1=1$.

Up to date, the following natural question has no complete answer.
\begin{quest}\label{Q1}
Which primitive graphs have equal Feynman period?
\end{quest}
In the context of four-dimensional $\phi^4$ theory, D. Broadhurst and D. Kreimer worked out two transformations that leave the Feynman period invariant:
Planar duality and conformal symmetry \cite{BK}.
Planar duality is based on a Fourier transformation of the propagators (the Fourier identity).
While the Fourier identity is somewhat exceptional, it was shown in \cite{Census} that the conformal symmetry has a deeper structure.
Similar to all massless $p$-integrals leading to the same Feynman period (using scale symmetry or projective geometry), it was shown that
conformal symmetry effectively adds a vertex `$\infty$' to the graph. This vertex $\infty$ connects to all external legs of the original graph.
The resulting `completion' $\overline{G}$ of $G$ is homogeneous in the sense that it is a vacuum graph in the underlying QFT. The completed graph, e.g.,
is four-regular in $\phi^4$ theory and three-regular in $\phi^3$ theory (every vertex has degree four or three, respectively).

In analogy to $p$-integrals, one proves that any decompletion (opening up a vertex $\infty$) of a completed graph gives the same Feynman period.
Therefore, after completion one forgets the label $\infty$
(as well as the labels $0$ and $1$) and lifts the Feynman period to a number that is assigned to the unlabeled vacuum graph $\overline{G}$. Every completed graph
can be considered as an equivalence class of graphs with equal Feynman period. The equivalence class consists of the choices of three `external' vertices
$0$, $1$, $\infty$ in $\overline{G}$.
Accordingly, we write
\begin{equation}
P_{\overline{G}}=P_G,\quad\text{if $G=\overline{G}\setminus v$ is any decompletion of $\overline{G}$.}
\end{equation}

For a completed graph, primitivity can be formalted in purely combinatorial terms.
\begin{defn}\label{compprim}
A graph $\overline{G}$ is weighted $D/\lambda$ regular if $N_v=D/\lambda$ for all vertices $v$ in $\overline{G}$; see (\ref{Nv}).
A weighted $D/\lambda$ regular graph $\overline{G}$ is weighted internally $n$-connected if the sum of the weights of cut edges is $\geq n$ for every cut that
does not separate a single vertex.
The weighted internal edge connectivity of $\overline{G}$ is $n$ if $\overline{G}$ is weighted internally $n$-connected but not weighted internally $m$-connected for $m>n$
(i.e.\ $n$ is the weight of the minimum non-trivial edge cut).

A weighted $D/\lambda$ regular graph is completed primitive if its weighed internal edge connectivty is $>D/lambda$.
\end{defn}

Note that for completed primitivity it is essential to consider internal cuts because in a $D/\lambda$ regular graph, a cut that separates off s single vertex always has
weight $D/\lambda$.

In Section \ref{sect3} we will prove the following proposition which is the weighted analog of Proposition 2.6 in \cite{Census}, see Section \ref{sect3}.
\begin{prop}\label{prop:conv1}
A weighted graph $G$ is primitive (i.e.\ the Feynman period (\ref{PG}) exists) if and only if its completion is completed primitive.
\end{prop}
The physical interpretation of the proposition is that any edge cut with total weight $\leq D/\lambda$ has a divergent insertions on both sides of the cut
and hence, after the deletion of $\infty$ on one side, a subdivergence on the other side.

In \cite{Census}, an entirely new transformation of completed primitive graphs was found: the twist,
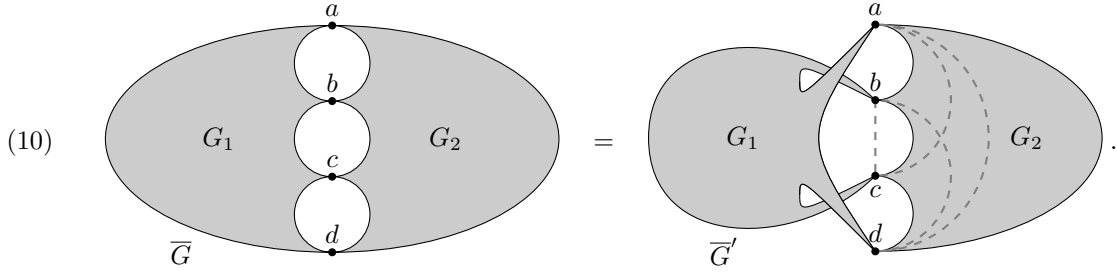
\begin{equation}\label{twist}
\begin{tikzpicture}[baseline={([yshift=-1.3ex]current bounding box.center)}]
    \coordinate[label=above:$a$] (va) at (0,3);
    \coordinate[label=above:$b$] (vb) at (0,2);
    \coordinate (vm) at (0,1.5);
    \coordinate[label=above:$c$] (vc) at (0,1);
    \coordinate[label=above:$d$] (vd) at (0,0);
    \filldraw (va) circle (1.3pt);
    \filldraw (vb) circle (1.3pt);
    \filldraw (vc) circle (1.3pt);
    \filldraw (vd) circle (1.3pt);
    \draw[fill=black!20] (va) arc (90:270:3 and 1.5) arc (270:90:.5) arc (270:90:.5) arc (270:90:.5);
    \draw[fill=black!20] (va) arc (90:-90:3 and 1.5) arc (-90:90:.5) arc (-90:90:.5) arc (-90:90:.5);
    \node (G1) at (-1.5,1.5) {$G_1$};
    \node (G2) at (+1.5,1.5) {$G_2$};
    \node (G) at (-2,0) {$\overline{G}$};
%
%
\end{tikzpicture}
\quad = \quad
\begin{tikzpicture}[baseline={([yshift=-1.3ex]current bounding box.center)}]
    \coordinate[label=above:$a$] (va) at (0,3);
    \coordinate[label=above:$b$] (vb) at (0,2);
    \coordinate (vm) at (0,1.5);
    \coordinate[label=below:$c$] (vc) at (0,1);
    \coordinate[label=above:$d$] (vd) at (0,0);
    \draw[fill=black!20] (vb) .. controls (-1,3) and (-3,3) .. (-3,1.5) .. controls (-3,0) and (-1,0) .. (vc) .. controls (-1,.5) .. (-1,.75) .. controls (-1,1) .. (vd) .. controls (-1,1.5) .. (va) .. controls (-1,2) .. (-1,2.25) .. controls (-1,2.5) .. (vb);
    \draw[fill=black!20] (-1,.75) .. controls (-1,1) .. (vd) .. controls (-1,1.5) .. (va) .. controls (-1,2) .. (-1,2.25);
    \draw[fill=black!20] (va) arc (90:-90:3 and 1.5) arc (-90:90:.5) arc (-90:90:.5) arc (-90:90:.5);
    \draw[dashed,black!50,thick] (vb) -- (vc) arc (-90:90:1) arc (90:-90:1.5) arc (-90:90:1);
    \node (G1) at (-1.75,1.5) {$G_1$};
    \node (G2) at (+2,1.5) {$G_2$};
    \node (G) at (-2,0) {$\overline{G}'$};
%
%
%
    \filldraw (va) circle (1.3pt);
    \filldraw (vb) circle (1.3pt);
    \filldraw (vc) circle (1.3pt);
    \filldraw (vd) circle (1.3pt);
\end{tikzpicture}\;.
\end{equation}

If $\overline{G}$ has a four-vertex split $a,b,c,d$ into $G_1$ and $G_2$, then $G_1$ can be twisted and edges can be added along the dashed
four-cycle $acbd$ such that the twisted graph $\overline{G}'$ is a completed Feynman graph. Then
\begin{equation}
P_{\overline{G}}=P_{\overline{G}'}.
\end{equation}
Originally, the twist was defined in four dimensions, but it generalizes to any dimension $D$ \cite{gfe}.

Four-vertex splits can be combined with planar duality to the Fourier split \cite{Furtherphi4}.

All known identities operate on completed graphs that have a nontrivial four-vertex split or a planar decompletion.
In this article, we prove a new identity on Feynman periods which neither requires a four-vertex split nor a planar decompletion.
This five-twist identity can be considered as a twist in a five-vertex cut of the completed graph.

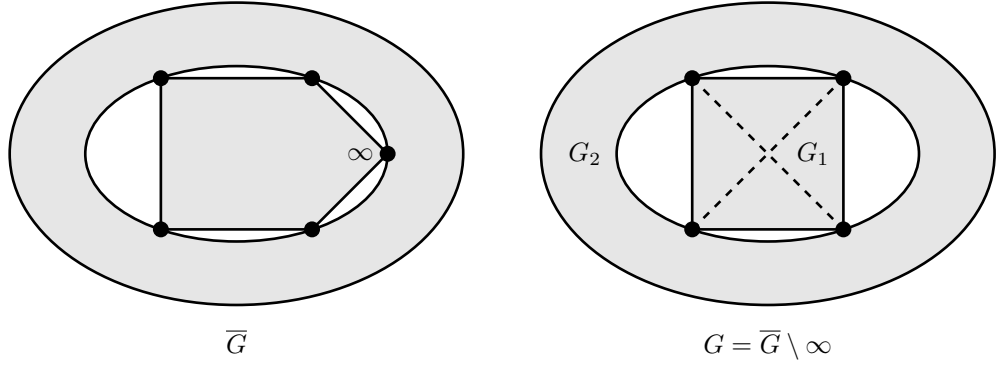
\begin{figure}
\begin{tikzpicture}
\begin{scope}[local bounding box=Gbar]
    \filldraw[li,fill=black!10] (0,0) ellipse (3 and 2);
    \filldraw[li,fill=white] (0,0) ellipse (2 and 1.16);
    \filldraw[li,fill=black!10] (2,0) -- (1,1) -- (-1,1) -- (-1,-1) -- (1,-1) -- (2,0);
    \fill (2,0) circle (3pt) node[anchor=east] {$\infty\,$};
    \fill (1,1) circle (3pt);
    \fill (-1,1) circle (3pt);
    \fill (1,-1) circle (3pt);
    \fill (-1,-1) circle (3pt);
    \node[below=0.2 of Gbar] {$\overline{G}$};
\end{scope}

\begin{scope}[xshift=200,local bounding box=G]
    \filldraw[li,fill=black!10] (0,0) ellipse (3 and 2);
    \filldraw[li,fill=white] (0,0) ellipse (2 and 1.16);
    \filldraw[li,fill=black!10] (-1,-1) rectangle (1,1);
    \draw[li,dashed] (-1,-1) -- (1,1);
    \draw[li,dashed] (-1,1) -- (1,-1);
    \fill (1,1) circle (3pt);
    \fill (-1,1) circle (3pt);
    \fill (1,-1) circle (3pt);
    \fill (-1,-1) circle (3pt);
    \node at (0.6,0) {$G_1$};
    \node at (-2.42,0) {$G_2$};
    \node[below=0.2 of G] {$G=\overline{G}\setminus\infty$};
\end{scope}
\end{tikzpicture}
\caption{The five-twist on the completed graph (left) and as reflections along diagonals of a square in the decompleted graph (right).
The shaded areas stand for any subgraphs. Only if the graph $G_1$ has specific properties, the five-twist becomes an identity for Feynman periods.}
\label{fig:5twist}
\end{figure}

We label one of the five split vertices $\infty$ and delete it from $\overline{G}$. The resulting decompletion $G$ decomposes along the four remaining split vertices
into $G_1$ and $G_2$. The split gives an edge bipartition of $G$ into $G_1$ and $G_2$ while the vertices of $G_1$ and $G_2$ share the split vertices; see Figure \ref{fig:5twist}.

Edges between the four split vertices can either belong to $G_1$ or to $G_2$ giving rise to different splits. Without restriction, however, we can assume
that $G_1$ has no such edge, so that $G_2$ is the subraph of $G$ that is induced by its vertices. Other cases do not lead to new identities.

We assume that neither $G_1$ nor $G_2$ is empty (otherwist the five-twist is trivial).
The graphs $G_2$, however, may only have the split vertices and the edges between the split vertices.
A non-trivial example of this type is depicted in Figure \ref{fig:P72} of Section \ref{sect3a}.

Feynman rules associate four-point functions to the split graphs. The idea is to transform the graph $G_1$ without changing
its four-point function and replace $G_1$ in $G$ by its transformation. Then, the Feynman period does not change.

In Section \ref{sect2}, we prove that an internally completed four-point function\footnote{A grpah is internally completed if its internal vertices have weighted degree $D/\lambda$.}
is invariant under a double transposition of its external vertices if (and only if) the degrees of the external vertices are stable under the transposition.
We cannot use this invariance directly because the only setup in which $G_1$ is internally completed is when $\infty$ does not connect to the interior of $G_1$.
In this case, the four split vertices in $G$ also give a four-vertex split of the completed graph $\overline{G}$ and we obtain the standard twist.

However, it can happen that $G_1$ is externally planar\footnote{A graph is externally planar if it has a planar embedding such that the external vertices are on the outer face.}.
Because any Fourier transform is invertible, the four-point integral of $G_1$ is determined by its Fourier transform.
Up to a constant factor, the Fourier transform of $G_1$  which is given by the four-point function of its planar dual $G_1^\ast$.
It may also happen that $G_1^\ast$ is internally completed. In $\phi^4$ theory, this is the case if (and only if) $G_1$ is a mesh of squares; see Figure \ref{fig:G1}.
In this setup, the four-point function of the dual $G_1^\ast$ is invariant under double transpositions that do not change the degrees of its external vertices.
This means that $G_1$ is reflected along one or both of its diagonals while the outer face of $G_1$ keeps the number of edges between the external
vertices.

If both $G_1$ and $G_2$ have vertices of degree three, the vertex $\infty$ in $\overline{G}$ connects to $G_1$ and to $G_2$. In this case, the five vertex split of $\overline{G}$
is not also a four-vertex split (by dropping $\infty$ from the split vertices). However, even in the case that $G_2$ only has vertices of degree four and the five vertex split
is also a four vertex split, the five-twist can give a new identity.

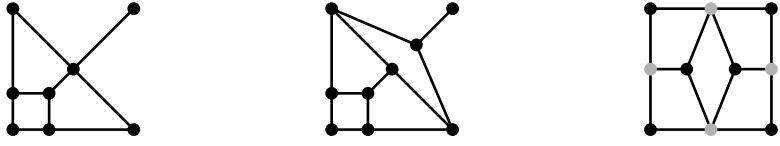
\begin{figure}
\begin{tikzpicture}[scale=0.8]
\begin{scope}[local bounding box=G11]
    \draw[li] (-1,-1) -- (1,-1) -- (-1,1) -- (-1,-1);
    \draw[li] (1,1) -- (-0.4,-0.4);
    \draw[li] (-0.4,-1) -- (-0.4,-0.4) -- (-1,-0.4);
    \fill (-1,-1) circle (3pt);
    \fill (-1,-0.4) circle (3pt);
    \fill (-1,1) circle (3pt);
    \fill (-0.4,-1) circle (3pt);
    \fill (-0.4,-0.4) circle (3pt);
    \fill (1,-1) circle (3pt);
    \fill (0,0) circle (3pt);
    \fill (1,1) circle (3pt);
\end{scope}

\begin{scope}[xshift=150,local bounding box=G12]
    \draw[li] (-1,-1) -- (1,-1) -- (-1,1) -- (-1,-1);
    \draw[li] (1,1) -- (0.4,0.4);
    \draw[li] (0,0) -- (-0.4,-0.4);
    \draw[li] (-0.4,-1) -- (-0.4,-0.4) -- (-1,-0.4);
    \draw[li] (-1,1) -- (0.4,0.4) -- (1,-1);
    \fill (-1,-1) circle (3pt);
    \fill (-1,-0.4) circle (3pt);
    \fill (-1,1) circle (3pt);
    \fill (-0.4,-1) circle (3pt);
    \fill (-0.4,-0.4) circle (3pt);
    \fill (1,-1) circle (3pt);
    \fill (0,0) circle (3pt);
    \fill (0.4,0.4) circle (3pt);
    \fill (1,1) circle (3pt);
\end{scope}

\begin{scope}[xshift=300,local bounding box=G13]
    \draw[li] (-1,-1) rectangle (1,1);
    \draw[li] (-1,0) -- (-0.4,0) -- (0,1) -- (0.4,0) -- (1,0);
    \draw[li] (-0.4,0) -- (0,-1) -- (0.4,0);
    \fill[fill=black!30] (-1,0) circle (3pt);
    \fill[fill=black!30] (0,1) circle (3pt);
    \fill[fill=black!30] (1,0) circle (3pt);
    \fill[fill=black!30] (0,-1) circle (3pt);
    \fill (-1,-1) circle (3pt);
    \fill (1,1) circle (3pt);
    \fill (-1,1) circle (3pt);
    \fill (1,-1) circle (3pt);
    \fill (-0.4,0) circle (3pt);
    \fill (0.4,0) circle (3pt);
\end{scope}
\end{tikzpicture}
\caption{Some small examples of insertions $G_1$ in the five-twist (see main text). The twist vertices are on the external square. Larger examples are in Figure \ref{fig:G1a}.}
\label{fig:G1}
\end{figure}

\begin{defn}\label{def:5twist}
Let $G$ be a primitive graph $D=2\lambda+2$ dimensions.
Assume there exist four vertices in $G$ such that the edges of $G$ split into $G_1$ and $G_2$ as depicted in Figure \ref{fig:5twist} and
\begin{enumerate}
\item the graph $G_1$ is planar with the cut vertices on the outer face,
\item for each internal face of $G_1$, the sum of the weights of its $N$ edges is $(N-2)D/2\lambda$, and
\item the total weight of the edges between external vertices on the outer face does not change under the reflection(s).
\end{enumerate}
Then, any graph that is obtained from $G$ by a reflections along one or both of the dashed diagonals is a five-twist of $G$.

Two completed primitive graphs are related by a five-twist if any of their decompletions are five-twists.
\end{defn}

Like all other identities, the five-twist does not alter the loop order; it acts inside a given loop order.
Moreover, the inverse of a five-twist is also a five-twist, so that five-twists generate a finite group action on
(completed) primitive graphs of a certain loop order.

The main result of this article is the following theorem.
\begin{thm}\label{thm:main}
If the primitive graphs $G$ and $G'$ are related by a five-twist, then their Feynman periods (\ref{PG}) are equal, $P_{G'}=P_G$.
\end{thm}

We will prove in Theorem \ref{thm:main} that the periods of two graphs are equal if they are five-twists.

Because a completed primitive graph can have many non-isomorphic decompletions, it is useful to consider the five-twist as a transformation of a completed graph
by the following steps,
\begin{enumerate}
\item decomplete,
\item four-vertex split,
\item dualize,
\item twist,
\item dualize,
\item four-vertex glue,
\item complete.
\end{enumerate}

The article is organized as follows.
In the next section we prove a formula for the double transposition of internally completed four-point integrals (these are conformal integrals in Super Yang-Mills theories \cite{SYM}).
Then, we prove the five-twist identity for scalar QFTs in Section \ref{sect3}.
In Section \ref{sect3a} we specialize to four-dimensional $\phi^4$ theory and can be read independently.
Those readers who are only interested in $\phi^4$ theory may skip Sections \ref{sect2} and \ref{sect3} and proceed directly to Section \ref{sect3a}.
It can also be useful to read Section \ref{sect3a} as an extended introduction before going to Sections \ref{sect2} and \ref{sect3}.
In Section \ref{sect4} we finally present exhaustive lists of five-twists until loop order eleven in $\phi^4$ theory and prove the independence of the five-twist for
the graph $P_{9,103}$.

Applications of the five-twist beyond $\phi^4$ theory have not yet been studied.
We emphasize that the five-twist is insufficient to (even experimentally) anwer Question \ref{Q1}.
The main purpose of this article is to show that simple ideas may lead to new identities. We hope to inspire the community to search for the missing transformation(s).
A better understanding of integral identities can have implications beyond Feynman periods because typically it is possible to extend the identities
to graphs with subdivergences or even to graphs in QFTs with a wider particle content.

\section*{Acknowlegements}
The author is supported by the DFG-grant SCHN 1240/3-1.

\section{Internally completed four-point integrals}\label{sect2}
Like most other known identities, the five-twist is proved in position space.
We consider Feynman graphs with four external vertices $z_0$, $z_1$, $z_2$, $z_3$, such that every other (internal) vertex
has weighted degree $D/\lambda$.

\begin{defn}
A graph $G$ with $N$ external vertices is internally completed if every other (internal) vertex has weighted degree $D/\lambda$.
To any external vertex $z_i$, $i=0,\ldots,N-1$, we associate a (position space) vector $z_i\in\RR^D$.
The position space Feynman integral of $G$ is
\begin{equation}\label{AG}
A_G(z_0,\ldots,z_{N-1})=\Big(\prod_v\int_{\RR^D}\frac{\dd^Dx_v}{\pi^{D/2}}\Big)\prod_{e\in\sE_G}p_e(x,z),
\end{equation}
where the integration is over the internal vertices $x_i\in\RR^D$ and the propagator $p_e$ depends on $z$ if $e$ connects to an external vertex.

For $i,j\in\{0,\ldots,N-1\}$ we define
\begin{equation}
z_{ij}=z_i-z_j.
\end{equation}
\end{defn}

The (convergent) Feynman integral of an internally completed four-point graph is given by a `graphical function' \cite{gf,gfe}.
Graphical functions are single-valued real-analytic functions on $\CC\backslash\{0,1\}$ which are defined by the Feynman integral of a three-point function.
We assume that the reader is familiar with the basic properties of graphical functions which are summarized in the first sections of \cite{gfe}.
In particular, it is costumary in the theory of graphical functions to replace the external labels $z_0,z_1,z_2,z_3\in\RR^D$ with $0,1,z,\infty$ that are also points
on the Riemann sphere $\CC\cup\{\infty\}$. With this identification, Equation (\ref{inv}) becomes tautological.

\begin{prop}\label{prop:AfG}
Let $G$ be an internally completed four-point graph. Let $G_{01z}=G\backslash\{3\}$ be $G$ after the deletion of vertex $z_3$ (together with its adjacent edges)
and label $z$ for $z_2$. Then, $G_{01z}$ is the graph of the graphical function $f_{G_{01z}}^{(\lambda)}(z)$ and
\begin{align}\label{AfG}
&A_G(z_0,z_1,z_2,z_3)\\\nonumber
&\quad=||z_{10}||^{\lambda(-N_0-N_1-N_2+N_3)}||z_{30}||^{\lambda(-N_0+N_1+N_2-N_3)}||z_{31}||^{\lambda(N_0-N_1+N_2-N_3)}||z_{32}||^{-2\lambda N_2}f_{G_{01z}}^{(\lambda)}(z),
\end{align}
where $z_0$, $z_1$, $z_2$, $z_3$ are related to $z\in\CC$ and its complex conjugate $\zz$ via the invariants
\begin{equation}\label{inv}
\frac{||z_{20}||^2||z_{31}||^2}{||z_{10}||^2||z_{32}||^2}=z\zz,\qquad \frac{||z_{21}||^2||z_{30}||^2}{||z_{10}||^2||z_{32}||^2}=(z-1)(\zz-1).
\end{equation}
\end{prop}
The Feynman integral $A_G(z_0,z_1,z_2,z_3)$ is regular if $z\in\CC\backslash\{0,1\}$.
\begin{proof}
By translational invariance, we have
$$
A_G(z_0,z_1,z_2,z_3)=A_G(z_{03},z_{13},z_{23},0).
$$
For every vector $0\neq x\in\RR^D$, we consider the inversion $x\mapsto\tilde x=x/||x||^2$. We obtain that, for the edge $e=xy$ with weight $\nu_e$, the propagator $p_e$ (Equation (\ref{eqpe}))
transforms according to
$$
p_e(\tilde x,\tilde y)=(||x||\,||y||)^{2\lambda\nu_e}p_e(x,y),\qquad p_e(\tilde x,0)=||x||^{2\lambda\nu_e}.
$$
Inverting the internal variable $x$ gives a factor of $||x||^{-2D}$ from the integration measure. This compensates the factor $||x||^{2\lambda D/\lambda}$ from the transformation
of the propagators. We conclude from the above equations that all propagators that connect to $0$ vanish
while the external variables $z_{ij}$ are transformed to $\tilde z_{ij}$ (because inversion is an involution).
We obtain
$$
A_G(z_{03},z_{13},z_{23},0)=||\tilde z_{03}||^{2\lambda N_0}||\tilde z_{13}||^{2\lambda N_1}||\tilde z_{23}||^{2\lambda N_2}A_{G\backslash\{3\}}(\tilde z_{03},\tilde z_{13},\tilde z_{23}).
$$
The three-point function on the right hand side is given by the graphical function $f_{G_{01z}}^{(\lambda)}(z)$, see Equation (16) in \cite{gfe},
$$
A_{G\backslash\{3\}}(\tilde z_{03},\tilde z_{13},\tilde z_{23})=||\tilde z_{13}-\tilde z_{03}||^{-2\lambda N_{G_{01z}}}f_{G_{01z}}^{(\lambda)}(z),
$$
where
\begin{equation}\label{eqNG}
N_{G_{01z}}=\Big(\sum_{e\in G_{01z}}\nu_e\Big)-\frac{D}{2\lambda}V^{\mathrm{int}}
\end{equation}
($V^{\mathrm{int}}$ is the number of internal vertices in $G_{01z}$) and
$$
\frac{||\tilde z_{23}-\tilde z_{03}||^2}{||\tilde z_{13}-\tilde z_{03}||^2}=z\zz,\qquad\frac{||\tilde z_{23}-\tilde z_{13}||}{||\tilde z_{13}-\tilde z_{03}||}=(z-1)(\zz-1).
$$
For $i,j\in\{0,1,2\}$ we have
$$
||\tilde z_{i3}-\tilde z_{j3}||^2=\frac1{||z_{i3}||^2}-2\frac{z_{i3}\cdot z_{j3}}{||z_{i3}||^2||z_{j3}||^2}+\frac1{||z_{j3}||^2}=\frac{(z_{j3}-z_{i3})^2}{||z_{i3}||^2||z_{j3}||^2}=
\frac{||z_{ji}||^2}{||z_{i3}||^2||z_{j3}||^2}.
$$
Hence
$$
\frac{||\tilde z_{23}-\tilde z_{03}||^2}{||\tilde z_{13}-\tilde z_{03}||^2}=\frac{||z_{20}||^2||z_{31}||^2}{||z_{10}||^2||z_{32}||^2},
\qquad\frac{||\tilde z_{23}-\tilde z_{13}||^2}{||\tilde z_{13}-\tilde z_{03}||^2}=\frac{||z_{21}||^2||z_{30}||^2}{||z_{10}||^2||z_{32}||^2}.
$$
By summing the weights of half-edges in $G$, we obtain (see (\ref{eqNG}))
$$
\frac{D}\lambda V^{\mathrm{int}}+N_0+N_1+N_2+N_3=2\sum_{e\in G}\nu_e=2\sum_{e\in G_{01z}}\nu_e+2N_3.
$$
This gives $N_{G_{01z}}=(N_0+N_1+N_2-N_3)/2$ and hence
$$
||\tilde z_{13}-\tilde z_{03}||^{-2\lambda N_{G_{01z}}}=\Big(\frac{||z_{10}||}{||z_{30}||\,||z_{31}||}\Big)^{\lambda(-N_0-N_1-N_2+N_3)}.
$$
Collecting the factors gives the desired result.

The last statement of the proposition follows from the theory of graphical functions, see \cite{par}.
\end{proof}

After completion, graphical functions are invariant under double transpositions of the external vertices $0$, $1$, $z$, $\infty$ (Theorem 3.20 \cite{gf} and the text thereafter,
Theorem 14 in \cite{gfe}).
This gives rise to an identity for internally completed four-point integrals.

\begin{prop}\label{prop:AA}
Let $G$ be an internally completed four-point graph. Then
\begin{equation}\label{AA}
A_G(z_0,z_1,z_2,z_3)=\bigg(\frac{||z_{31}||}{||z_{20}||}\bigg)^{\lambda(N_0-N_1+N_2-N_3)}\bigg(\frac{||z_{21}||}{||z_{30}||}\bigg)^{\lambda(N_0-N_1-N_2+N_3)}A_G(z_1,z_0,z_3,z_2).
\end{equation}
\end{prop}
\begin{proof}
The Feynman integral on the right hand side of (\ref{AA}) can be interpreted as $A_{G'}(z_0,z_1,z_2,z_3)$ for a graph $G'$ that is $G$ with swapped labels $0,1$ and $2,3$.
Note that double transpositions keep the connection between $z_0,z_1,z_2,z_3$ and $z$ in (\ref{inv}).
We use Proposition \ref{prop:AfG} to convert (\ref{AA}) into an identity for graphical functions.
$$
f_{G_{01z}}^{(\lambda)}(z)=(||z_{10}||\,||z_{32}||)^{2\lambda(N_2-N_3)}(||z_{20}||\,||z_{31}||)^{\lambda(-N_0+N_1-N_2+N_3)}(||z_{21}||\,||z_{30}||)^{\lambda(N_0-N_1-N_2+N_3)}
f_{G'_{01z}}^{(\lambda)}(z).
$$
With (\ref{inv}), this becomes
$$
f_{G_{01z}}^{(\lambda)}(z)=(z\zz)^{\lambda(-N_0+N_1-N_2+N_3)/2}((z-1)(\zz-1))^{\lambda(N_0-N_1-N_2+N_3)/2}f_{G'_{01z}}^{(\lambda)}(z).
$$
To connect the graphical function of $G_{01z}$ to $G'_{01z}$, we complete the graph $G_{01z}$. This adds edges $z\infty$, $01$, $0\infty$, and $1\infty$ such that the external
vertices have degree zero. We get
$$
\nu_{z\infty}=-N_2,\;\nu_{01}=\frac{-N_0-N_1-N_2+N_3}2,\;\nu_{0\infty}=\frac{-N_0+N_1+N_2-N_3}2,\;\nu_{1\infty}=\frac{N_0-N_1+N_2-N_3}2.
$$
The completion $\overline{G}_{01z\infty}$ of $G_{01z}$ is invariant under double transpositions,
$$
f_{\overline{G}_{01z\infty}}^{(\lambda)}(z)=f_{\overline{G}_{10\infty z}}^{(\lambda)}(z).
$$
Decompleting $\overline{G}_{10\infty z}$ gives
$$
f_{\overline{G}_{10\infty z}}^{(\lambda)}(z)=((z-1)(\zz-1))^{-2\lambda\nu_{0\infty}}(z\zz)^{-2\lambda\nu_{1\infty}}f_{G'_{01z}}^{(\lambda)}(z).
$$
Inserting the weights $\nu_{0\infty}$ and $\nu_{1\infty}$ gives the result.
\end{proof}

\begin{cor}\label{cor:doublet}
Let $G$ be an internally completed four-point graph with $N_0=N_1$ and $N_2=N_3$. Then $A_G(z_0,z_1,z_2,z_3)$ is invariant under a double transposition of $z_0,z_1$ and $z_2,z_3$,
\begin{equation}\label{AA1}
A_G(z_0,z_1,z_2,z_3)=A_G(z_1,z_0,z_3,z_2).
\end{equation}
\end{cor}
\begin{proof}
This is an immediate consequence of (\ref{AA}).
\end{proof}

\begin{remark}\label{remark:doublet}
By permutation symmetry in $z_0$, $z_1$, $z_2$, $z_3$, we obtain from Proposition \ref{prop:AA} formulae for double transpositions $z_0\leftrightarrow z_2$, $z_1\leftrightarrow z_3$ and
$z_0\leftrightarrow z_3$, $z_1\leftrightarrow z_2$.
\begin{align}\label{AA2}
A_G(z_0,z_1,z_2,z_3)&=\bigg(\frac{||z_{32}||}{||z_{10}||}\bigg)^{N_0+N_1-N_2-N_3}\bigg(\frac{||z_{21}||}{||z_{30}||}\bigg)^{N_0-N_1-N_2+N_3}A_G(z_2,z_3,z_0,z_1)\\\nonumber
A_G(z_0,z_1,z_2,z_3)&=\bigg(\frac{||z_{32}||}{||z_{10}||}\bigg)^{N_0+N_1-N_2-N_3}\bigg(\frac{||z_{31}||}{||z_{20}||}\bigg)^{N_0-N_1+N_2-N_3}A_G(z_3,z_2,z_1,z_0).
\end{align}
Likewise, we get
\begin{align}\label{AA3}
A_G(z_0,z_1,z_2,z_3)=A_G(z_2,z_3,z_0,z_1)&,\quad\text{if\quad$N_0=N_2$ and $N_1=N_3$},\\\nonumber
A_G(z_0,z_1,z_2,z_3)=A_G(z_3,z_2,z_1,z_0)&,\quad\text{if\quad$N_0=N_3$ and $N_1=N_2$}.
\end{align}
\end{remark}

\section{Proofs}\label{sect3}
We first prove Proposition \ref{prop:conv1}.
\begin{proof}
Because the existence of the period (\ref{PG}) does not depend on the choice of the vertices $0,1,\infty$, we can fix an (arbitrary) choice of $0,1,\infty$ in the vertices
of $\overline{G}$.

Edge cuts of the completed graph $\overline{G}$ into $G_1$ and $G_2$ are in one to one correspondence with partitions of the vertex set of $\overline{G}$
into a subset $\sV$ and its complement.
By symmetry we may choose $G_1$ to be the induced subgraph $\overline{G}[\sV]$ (i.e.\ the subgraph with all edges of $\overline{G}$ whose vertices are in $\sV$).
Let $w=\sum_{\mathrm{cut}\,e}\nu_e$ be the weight of the cut. Adding the weights of half-edges, we get the identity
$$
w+2\sum_{e\in\overline{G}[\sV]}\nu_e=\frac{D}\lambda|\sV|=\frac{D}\lambda(|\sV^{\mathrm{int}}|+|\sV^{\mathrm{ext}}|),
$$
where we took into account that some of the vertices in $\sV$ can be external, i.e.\ $0$, $1$, or $\infty$. With
$$
N_{\overline{G}[\sV]}=\Big(\sum_{e\in\overline{G}[\sV]}\nu_e\Big)-\frac{D}{2\lambda}|\sV^{\mathrm{int}}|
$$
we obtain
$$
2N_{\overline{G}[\sV]}=\frac{D}{\lambda}|\sV^{\mathrm{ext}}|-w.
$$
Because $G_2$ has at least two vertices, we can locate the external vertices such that $|\sV^{\mathrm{ext}}|\leq1$.

To prove the convergence of $P_{\overline{G}\backslash\{\infty\}}$ in (\ref{PG}), it is convenient to use the more general convergence theorem for graphical functions.
To do this, we consider the graphical function of $\overline{G}$ with an additional isolated vertex $z$.
Because $z$ is isolated, the graphical function $f_{\overline{G}\cup\{z\}}^{(\lambda)}(z)=P_{\overline{G}\backslash\{\infty\}}$ is constant (see (\ref{AG})).

From Proposition 11 of \cite{gfe} we obtain that $f_{\overline{G}\cup\{z\}}^{(\lambda)}$ exists if and only if
$$
N_{(\overline{G}\cup\{z\})[\sV]}=N_{\overline{G}[\sV]}<(|\sV^{\mathrm{ext}}|-1)\frac{D}{2\lambda}
$$
for all vertex subsets $\sV$ with $|\sV^{\mathrm{ext}}|\leq1$. This condition becomes
$$
\frac{D}{\lambda}|\sV^{\mathrm{ext}}|-w<(|\sV^{\mathrm{ext}}|-1)\frac{D}{\lambda}
$$
which is equivalent to $w>D/\lambda$.
\end{proof}

Now we prove the five-twist identity in Theorem \ref{thm:main}.
\begin{proof}
The Feynman integral $A_{G_1}(x_1,x_2,x_3,x_4)$ of the insertion $G_1$ (where the $x_i$ are the split vertices)
is determined by its Fourier transform $A_{G_1}^\ast(p_1,p_2,p_3,p_4)$.
Momentum conservation provides a $D$-dimensional Dirac $\delta$ function $\delta^D(p_1+p_2+p_3+p_4)$.
For the coefficient of the $\delta$ function we use the coordinates $p_1=z_1-z_0$, $p_2=z_2-z_1$, $p_3=z_3-z_2$, $p_4=z_0-z_3$.
This determines the Fourier transform (up to the $\delta$ function and a constant) as the position space Feynman integral of the planar dual graph
$G_1^\ast$ with edge weighs $\nu_e^\ast=D/2\lambda-\nu_e$ (see e.g.\ Theorem 1.9 in \cite{par}) whose external vertices are labeled $z_i$, corresponding to the
chain from $x_i$ to $x_{i+1}$ on the outer face (where $x_0=x_4$).

The graph $G_1^\ast$ is internally completed because for every internal vertex $x$ that corresponds to a face in $G_1$ with $N$ edges,
the sum of the weights of adjacent edges $e\sim x$ is
$$
\sum_{e\sim x}\nu_e^*=\sum_{e\in\mathrm{face}\,x}\Big(\frac{D}{2\lambda}-\nu_e\Big)=\frac{ND}{2\lambda}-\frac{(N-2)D}{2\lambda}=\frac{D}\lambda.
$$
From Corollary \ref{cor:doublet} or from Remark \ref{remark:doublet}, we get that the Feynman integral $A_{G_1^\ast}(z_0,z_1,z_2,z_3)$
is invariant under a double transposition of its arguments if the degrees of the swapped vertices do not change.
A double transposition in $G_1^\ast$ becomes a reflection along one or both diagonals in the dual graph $G_1$.
The degrees of the external vertices in $G_1^\ast$ are determined by the sum of the edge weights on the corresponding side in the outer face of $G_1$.
This proves the theorem.
\end{proof}

\section{Periods in $\phi^4$ theory}\label{sect3a}

In $\phi^4$ theory every internal vertex has four edges of weight $\nu_e=1$. The dimension is $D=4$, $\lambda=1$ and the parametric expression for the period
$P_G$ of a truncated four-point graph $G$ is (see (\ref{PsiG}) and (\ref{para}))
\begin{equation}
P_G=\int_{\alpha_e>0}\frac{\Omega}{\Psi_G(\alpha)^2}.
\end{equation}

The first primitive graph is the bubble which has loop order one
(the independent cycles $h_G$ of a graph $G$). There exists no primitive $\phi^4$ graph with two loops and one primitive $\phi^4$ graph with three loops;
see Figure \ref{fig:bubbleK4}.

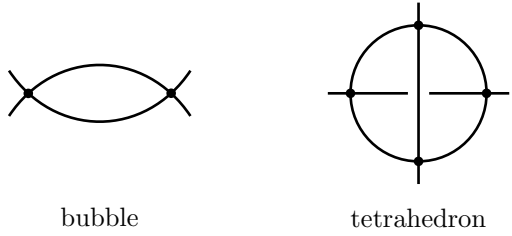
\begin{figure}
\begin{tikzpicture}[scale=0.6]
\begin{scope}[local bounding box=bubble]
	\coordinate (A) at (2,0.5);
	\coordinate (B) at (-2,0.5);
	\coordinate (C) at (-2,-0.5);
	\coordinate (D) at (2,-0.5);
    \draw[li,name path=li1] (A) .. controls (1,-1) and (-1,-1) .. (B);
    \draw[li,name path=li2] (C) .. controls (-1,1) and (1,1) .. (D);
    \fill[name intersections={of=li1 and li2}]
        (intersection-1) circle (3pt)
        (intersection-2) circle (3pt);
    \node[below=0.79 of bubble] {bubble};
\end{scope}

\begin{scope}[xshift=200,local bounding box=tetra]
	\coordinate (A) at (1.5,0);
	\coordinate (B) at (0,1.5);
	\coordinate (C) at (-1.5,0);
	\coordinate (D) at (0,-1.5);
	\coordinate (O) at (0,0);
    \draw[li] (A) -- (C);
	\draw[white,line width=8pt] (B) -- (D);
    \draw[li] (O) circle (1.5);
    \draw[li] (A) -- (2,0);
    \draw[li] (B) -- (0,2);
    \draw[li] (C) -- (-2,0);
    \draw[li] (D) -- (0,-2);
    \draw[li] (B) -- (D);
    \fill (A) circle (3pt);
    \fill (B) circle (3pt);
    \fill (C) circle (3pt);
    \fill (D) circle (3pt);
    \node[below=0.2 of tetra] {tetrahedron};
\end{scope}
\end{tikzpicture}
\caption{The bubble and the tetrahedron are the smallest primitive graphs in $\phi^4$ theory.}
\label{fig:bubbleK4}
\end{figure}

The Feynman periods of the bubble and the tetrahedron in Figure \ref{fig:bubbleK4} are $1$ and $6\zeta(3)=6\sum_{k=1}^\infty k^{-3}$, respectively.

In the wake of the visionary work by D. Broadhurst and D. Kreimer \cite{BK}, Feynman periods became a prominent topic in mathematics and in physics.
Based on the combinatorics of the graph polynomial in the parametric representation (\ref{para}), in \cite{BEK} a mathematical Feynman motive was defined for certain graphs $G$.
With the theory of graphical functions \cite{gf,gfe,Shlog}, it was possible to extend the data in \cite{BK} and in \cite{Census} to hundreds of graphs up to loop order
eleven in $\phi^4$ theory. In six-dimensional $\phi^3$ theory, results exist up to loop order nine. With this data, a connection to motivic Galois theory became visible \cite{coaction}
which led to further investigations of the motivic structure of QFTs \cite{Bcoact1,Bcoact2} (the `cosmic' Galois group) and of the geometries that underlie the number content
of $\phi^4$ periods \cite{SchnetzFq,BSmod,Sc2}.

It also became possible to prove (assuming mathematical standard conjectures) that Feynman periods are not always given by multiple zeta values (higher depths analogs
of the Riemann zeta function) \cite{K3,Lproofs}. Moreover, the entire zigzag family of Feynman periods (whose first member is the tetrahedron)
could be calculated \cite{ZZ,ZZ2,ZZ3}.

In 2019, E. Panzer conjectured that a combinatorial invariant, the Hepp bound, can identify equal Feynman periods in $\phi^4$ theory \cite{EPHepp}
(i.e.\ $P_{G_1}=P_{G_2}$ if and only if the Hepp bounds of $G_1$ and $G_2$ are equal).
This conjecture is supported by numerical evidence which became available with M. Borinsky's tropical Monte Carlo integration method \cite{MBtropical}.
Tropical geometry is also the basis for recent numerical calculations of Feynman periods to very high loop orders \cite{BFlogGAMMA,BTropsamp,PHBStatistics,PHBKSPredicting,PHBJTPrimAsymp}.

Some years after the Hepp bound, E. Panzer and K. Yeats found the Martin sequence which is an infinite family of combinatorial invariants that are associated
to Feynman graphs in scalar QFTs \cite{PYMartin}.
For a primitive $\phi^4$ graph, the Martin sequence is proved to be invariant under all known identities of the period. Every known combinatorial invariant of $\phi^4$ periods
can be derived from the Martin sequence, so that the Martin sequence serves as a unified theory of $\phi^4$ invariants (see Section 1.5 of \cite{PYMartin}).
It has been shown in several examples that the Feynman period itself can be obtained from the Martin sequence \cite{PTalk}.

These recent developments support the picture that there exist more identities on $\phi^4$ periods than those that can be explained by known transformations.
The first examples are at loop order eight, where it is conjectured that
\begin{equation}\label{conjid}
P_{8,30}=P_{8,36}\qquad\text{and}\qquad P_{8,31}=P_{8,35}
\end{equation}
in the notation of \cite{Census}.
The graphs of the periods $P_{8,35}$ and $P_{8,36}$ have neither a four-vertes split nor a planar decompletion, so that they don't transform under known identities.
At nine loops and beyond, there exist many more conjectured identities of this type.

Regretfully, the five-twist is not powerful enough to explain the conjectured identities (\ref{conjid}).
Still, starting at loop order eight, it gives new relations for $\phi^4$ periods (see Section \ref{sect4}). In most cases,
$\phi^4$ periods are connected to Feynman periods of graphs that are not in $\phi^4$ theory. It seems possible that the five-twist is more powerful outside $\phi^4$ theory.
In this case it may be interesting to see if one gets more results in combination with the existing identities that map $\phi^4$ periods to non-$\phi^4$
periods.

\section{Results and conclusions}\label{sect4}
The smallest nontrivial five-twist identity occurs at seven loops, where the Feynman period $P_{7,2}$ is connected to the non-$\phi^4$ period $P^{\mathrm{non}\,\phi^4}_{7,17}$
(in the numbering of \cite{Shlog}); see Figure \ref{fig:P72}.

\begin{figure}
\begin{tikzpicture}[scale=0.6]
\begin{scope}[local bounding box=P7_2]
    \draw[li] (-1.4,-1.4) -- (2,-2);
    \draw[white,line width=8pt] (-0.8,-0.8) -- (-0.8,-2);
    \draw[li] (-0.8,-2) -- (-0.8,-0.8) -- (-2,-0.8);
    \draw[li] (2,2) -- (-1.4,-1.4);
    \draw[li] (-2,2) .. controls (-1,3) and (1,3) .. (2,2);
    \draw[li] (-2,-2) .. controls (-3,-1) and (-3,1) .. (-2,2);
    \draw[li] (-2,-2) -- (2,-2) -- (-2,2) -- (-2,-2);
    \draw[li] (-2,-2) .. controls (1.3,-5) and (5,-1.3) .. (2,2);
    \draw[li] (2,-2) .. controls (3,-1) and (3,1) .. (2,2);
    \draw[li] (-2,-0.8) -- (-0.8,-2);
    \filldraw[black] (-2.1,-2.1) rectangle (-1.9,-1.9);
    \fill (-2,-0.8) circle (3pt);
    \filldraw[black] (-2.1,1.9) rectangle (-1.9,2.1);
    \fill (-0.8,-2) circle (3pt);
    \fill (-0.8,-0.8) circle (3pt);
    \fill[fill=black!30] (-1.4,-1.4) circle (3pt);
    \filldraw[black] (1.9,-2.1) rectangle (2.1,-1.9);
    \fill (0,0) circle (3pt);
    \filldraw[black] (1.9,1.9) rectangle (2.1,2.1);
\end{scope}

\begin{scope}[xshift=230,local bounding box=Q7_17]
    \draw[li] (1.4,1.4) -- (2,-2);
    \draw[li] (-2,-2) .. controls (1,-1) and (1,0) .. (1.4,1.4);
    \draw[white,line width=8pt] (0.8,0.8) -- (2,0.8);
    \draw[white,line width=8pt] (0,0) -- (1.5,-1.5);
    \draw[li] (0.8,2) -- (0.8,0.8) -- (2,0.8);
    \draw[li] (-2,-2) -- (1.4,1.4);
    \draw[li] (-2,2) .. controls (-1,3) and (1,3) .. (2,2);
    \draw[li] (-2,-2) .. controls (-3,-1) and (-3,1) .. (-2,2);
    \draw[li] (2,2) -- (2,-2) -- (-2,2) -- (2,2);
    \draw[li] (-2,-2) .. controls (1.3,-5) and (5,-1.3) .. (2,2);
    \draw[li] (2,-2) .. controls (3,-1) and (3,1) .. (2,2);
    \draw[li] (2,0.8) -- (0.8,2);
    \draw[dashed] (1.4,1.4) -- (2,2);
    \filldraw[black] (-2.1,-2.1) rectangle (-1.9,-1.9);
    \fill (2,0.8) circle (3pt);
    \filldraw[black] (-2.1,1.9) rectangle (-1.9,2.1);
    \fill (0.8,2) circle (3pt);
    \fill (0.8,0.8) circle (3pt);
    \fill[fill=black!30] (1.4,1.4) circle (3pt);
    \filldraw[black] (1.9,-2.1) rectangle (2.1,-1.9);
    \fill (0,0) circle (3pt);
    \filldraw[black] (1.9,1.9) rectangle (2.1,2.1);
\end{scope}
    \node[below=-1 of P7_2] {$P_{7,2}$};
    \node[below=-1 of Q7_17] {$P^{\mathrm{non}\,\phi^4}_{7,17}$};
    \node at (4.2,0) {$=$};
\end{tikzpicture}
\caption{The smallest nontrivial five-twist links $P_{7,2}$ to the non-$\phi^4$ period $P^{\mathrm{non}\,\phi^4}_{7,17}$.
The dashed edge on the right hand side has weight $-1$ (a numerator edge), all other edges have weight $1$.
The four twist vertices are plotted as squares and $\infty$ is the gray vertex (see also the first graph in Figure \ref{fig:G1} for the insertion $G_1$).}
\label{fig:P72}
\end{figure}
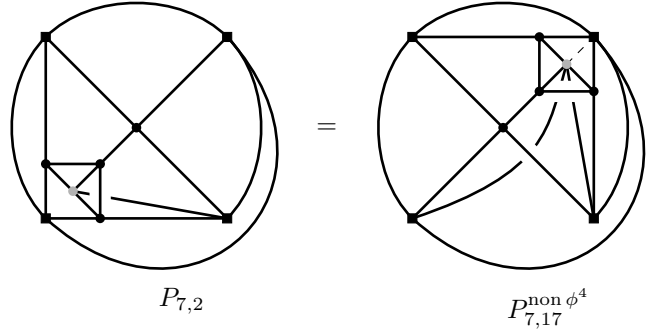

The five-twist establishes an infinite family of subgraphs $G_1$ that fulfill all conditions for a transformation.
In $\phi^4$ theory, the restriction of this family to a given maximum number of vertices is rather small.
This is because, in addition to the restrictions which are explained above, the graph $G_1$ must not be symmetric under the reflection along an admissible diagonal.
Moreover, in $\phi^4$ theory, the graph $G_1$ must not have more than four vertices of degree three.
Otherwise, it connects to $\infty$ with more than four edges which brings $\overline{G}$ outside $\phi^4$ theory (the vertex $\infty$ must have an edge of negative weight
to compensate for the $>4$ edges that connect to $G_1$). The smallest admissible $G_1$ is depicted on the left of Figure \ref{fig:G1}.
The middle graph does not directly give rise to a transformation of a $\phi^4$ graph because it has five internal vertices of degree three.
It still may give a transformation of a $\phi^4$ period if it sits in a chain of identities which lead outside $\phi^4$ theory and then back into $\phi^4$ theory again.

In general, combinations with other transformations are important:
It may happen that one obtains a transformation of a $\phi^4$ period by first mapping it to a non-$\phi^4$ period using a
known identity and then applying the five-twist to the non-$\phi^4$ graph. If the resulting graph is non-$\phi^4$, then this transformation is
lost by the restriction to $\phi^4$ graphs. If the non-$\phi^4$ graph is connected by another known identity to a $\phi^4$ graph,
it may happen that identities between $\phi^4$ graphs are lost by the restriction of the five-twist to $\phi^4$ graphs.

Nevertheless, as a first step towards understanding the five-twist, in this section we restrict ourselves to applying the five-twist directly
to $\phi^4$ graphs.

In fact, all known graphs $G_1$ that can be used to directly transform a $\phi^4$ graph, also lead to a (different) nontrivial four vertex split in the completed graph $\overline{G}$.
So, in all known examples, the five-twist is not the only transformation of $\overline{G}$. It can happen, that a twist on some four vertex split
is equal to the five-twist. The right hand side of Figure \ref{fig:G1} is such an example where the vertices of the standard twist are gray, see Section \ref{sect4}.

Identities (both inside $\phi^4$ theory and from $\phi^4$ theory to non-$\phi^4$ graphs) are rather rare compared to the twist (\ref{twist}):
Inside $\phi^4$ theory, the twist has one identity at seven loops and eight identities at eight loops (Table 4 in \cite{Census}) while the first
non-trivial five-twists inside $\phi^4$ theory emerge at nine loops (see (\ref{5ids})).

Up to loop order eight, we get the following identities between $\phi^4$ periods and non-$\phi^4$ periods (in the notation of \cite{Census,Shlog}).
$$
P_{7,2}=P^{\mathrm{non}\,\phi^4}_{7,17},\;P_{8,6}=P^{\mathrm{non}\,\phi^4}_{8,149},\;P_{8,14}=P^{\mathrm{non}\,\phi^4}_{8,460},\;
P_{8,18}=P^{\mathrm{non}\,\phi^4}_{8,75},\;P_{8,19}=P^{\mathrm{non}\,\phi^4}_{8,150},\;P_{8,20}=P^{\mathrm{non}\,\phi^4}_{8,379}.
$$
The first identity is depicted in Figure \ref{fig:P72}. The first and the last identities can also be explained as twists (\ref{twist}) (on sets
of vertices that differ from the five-twist). The identites two to five are not twist identities. 

Beyond eight loops, we have no list of non-$\phi^4$ graphs. In this case, we only looked for identities inside $\phi^4$ theory without detours via intermediate non-$\phi^4$ graphs.
By exhaustive search, we obtain
\begin{align}\label{5ids}
&P_{9,78}=P_{9,93},\quad P_{9,158}=P_{9,160},\nonumber\\
&P_{10,225}=P_{10,283},\quad P_{10,227}=P_{10,284},\quad P_{10,553}=P_{10,554},\quad P_{10,867}=P_{10,912},\nonumber\\
&P_{11,269}=P_{11,338},\quad P_{11,271}=P_{10,339},\quad P_{11,924}=P_{11,965},\quad P_{11,926}=P_{11,1072},\nonumber\\
&P_{11,928}=P_{11,1073},\quad P_{11,967}=P_{11,1076},\quad P_{11,969}=P_{11,1083},\quad P_{11,1117}=P_{10,1121},\nonumber\\
&P_{11,2930}=P_{11,2955},\quad P_{11,3879}=P_{11,3880},\quad P_{11,3881}=P_{11,3882},\quad P_{11,3884}=P_{11,3885}.
\end{align}
All five-twist identities up to loop order eleven can also be obtained from standard twists (\ref{twist}). So, from direct use of the five-twist we obtain no new identity for
$\phi^4$ periods up to loop order eleven.
For some graphs $G_1$, this is easily explained: If e.g.\ the graph $G_1$ is the right graph in Figure \ref{fig:G1}, then
the three horizontal vertices in the middle have degree three and connect (after completion) to the vertex $\infty$. If the full graph $\overline{G}$ is in $\phi^4$ theory,
then $\infty$ connects to no other vertices, so that the gray vertices give a four-vertex split of $\overline{G}$. A nontrivial twist on these vertices (which all have
degree two after completion) is identical to a reflection along a diagonal of the external square.

Starting from loop order eight, there exist five-twists linking $\phi^4$ periods to non-$\phi^4$ periods that cannot be obtained as standard twists at any vertices.
By exhaustive search, all such five-twists up to eleven loops belong to one family of insertions which is obtained from the left graph in Figure \ref{fig:G1} by
adding edges and vertices. The first members of this family are depicted in Figure \ref{fig:G1a}, where the external vertices are black squares.
Note that only the first graph has less than four vertices of degree three. In all other graphs, $\infty$ connects to four points in the insertion, so that
for $\phi^4$ graphs the four external vertices also give a four-vertex split. The five-twist still differs from the standard twist, because the latter is given by reflections
whose axes cut the faces of the external square. In this case, the five-twist and the standard twist generate the full dihedral group $D_4$ of the square.

It would be interesting to analyse this family of insertions combinatorially. It is unclear if more families of insertions exist that are non-trivial in the sense that
they give transformations that are not also standard twists.

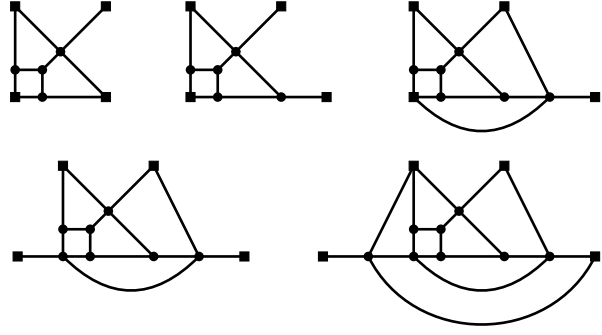
\begin{figure}
\begin{tikzpicture}[scale=0.6]
\begin{scope}[local bounding box=G11]
    \draw[li] (-1,-1) -- (1,-1) -- (-1,1) -- (-1,-1);
    \draw[li] (1,1) -- (-0.4,-0.4);
    \draw[li] (-0.4,-1) -- (-0.4,-0.4) -- (-1,-0.4);
    \filldraw[black] (-1.1,-1.1) rectangle (-0.9,-0.9);
    \filldraw[black] (-1.1,0.9) rectangle (-0.9,1.1);
    \filldraw[black] (0.9,-1.1) rectangle (1.1,-0.9);
    \filldraw[black] (0.9,0.9) rectangle (1.1,1.1);
    \fill (-1,-0.4) circle (3pt);
    \fill (-0.4,-1) circle (3pt);
    \fill (-0.4,-0.4) circle (3pt);
    \fill (0,0) circle (3pt);
\end{scope}

\begin{scope}[xshift=110,local bounding box=G12]
    \draw[li] (-1,-1) -- (1,-1) -- (-1,1) -- (-1,-1);
    \draw[li] (1,1) -- (-0.4,-0.4);
    \draw[li] (-0.4,-1) -- (-0.4,-0.4) -- (-1,-0.4);
    \draw[li] (1,-1) -- (2,-1);
    \filldraw[black] (-1.1,-1.1) rectangle (-0.9,-0.9);
    \filldraw[black] (-1.1,0.9) rectangle (-0.9,1.1);
    \filldraw[black] (1.9,-1.1) rectangle (2.1,-0.9);
    \filldraw[black] (0.9,0.9) rectangle (1.1,1.1);
    \fill (1,-1) circle (3pt);
    \fill (-1,-0.4) circle (3pt);
    \fill (-0.4,-1) circle (3pt);
    \fill (-0.4,-0.4) circle (3pt);
    \fill (0,0) circle (3pt);
\end{scope}

\begin{scope}[xshift=250,local bounding box=G13]
    \draw[li] (-1,-1) -- (1,-1) -- (-1,1) -- (-1,-1);
    \draw[li] (1,1) -- (-0.4,-0.4);
    \draw[li] (-0.4,-1) -- (-0.4,-0.4) -- (-1,-0.4);
    \draw[li] (1,-1) -- (3,-1);
    \draw[li] (1,1) -- (2,-1);
    \draw[li] (-1,-1) .. controls (0,-2) and (1,-2) .. (2,-1);
    \filldraw[black] (-1.1,-1.1) rectangle (-0.9,-0.9);
    \filldraw[black] (-1.1,0.9) rectangle (-0.9,1.1);
    \filldraw[black] (2.9,-1.1) rectangle (3.1,-0.9);
    \filldraw[black] (0.9,0.9) rectangle (1.1,1.1);
    \fill (2,-1) circle (3pt);
    \fill (1,-1) circle (3pt);
    \fill (-1,-0.4) circle (3pt);
    \fill (-0.4,-1) circle (3pt);
    \fill (-0.4,-0.4) circle (3pt);
    \fill (0,0) circle (3pt);
\end{scope}

\begin{scope}[xshift=30,yshift=-100,local bounding box=G14]
    \draw[li] (-2,-1) -- (1,-1) -- (-1,1) -- (-1,-1);
    \draw[li] (1,1) -- (-0.4,-0.4);
    \draw[li] (-0.4,-1) -- (-0.4,-0.4) -- (-1,-0.4);
    \draw[li] (1,-1) -- (3,-1);
    \draw[li] (1,1) -- (2,-1);
    \draw[li] (-1,-1) .. controls (0,-2) and (1,-2) .. (2,-1);
    \filldraw[black] (-2.1,-1.1) rectangle (-1.9,-0.9);
    \filldraw[black] (-1.1,0.9) rectangle (-0.9,1.1);
    \filldraw[black] (2.9,-1.1) rectangle (3.1,-0.9);
    \filldraw[black] (0.9,0.9) rectangle (1.1,1.1);
    \fill (-1,-1) circle (3pt);
    \fill (2,-1) circle (3pt);
    \fill (1,-1) circle (3pt);
    \fill (-1,-0.4) circle (3pt);
    \fill (-0.4,-1) circle (3pt);
    \fill (-0.4,-0.4) circle (3pt);
    \fill (0,0) circle (3pt);
\end{scope}

\begin{scope}[xshift=250,yshift=-100,local bounding box=G14]
    \draw[li] (-3,-1) -- (1,-1) -- (-1,1) -- (-1,-1);
    \draw[li] (1,1) -- (-0.4,-0.4);
    \draw[li] (-0.4,-1) -- (-0.4,-0.4) -- (-1,-0.4);
    \draw[li] (1,-1) -- (3,-1);
    \draw[li] (1,1) -- (2,-1);
    \draw[li] (-2,-1) -- (-1,1);
    \draw[li] (-1,-1) .. controls (0,-2) and (1,-2) .. (2,-1);
    \draw[li] (-2,-1) .. controls (-1,-3) and (2,-3) .. (3,-1);
    \filldraw[black] (-3.1,-1.1) rectangle (-2.9,-0.9);
    \filldraw[black] (-1.1,0.9) rectangle (-0.9,1.1);
    \filldraw[black] (2.9,-1.1) rectangle (3.1,-0.9);
    \filldraw[black] (0.9,0.9) rectangle (1.1,1.1);
    \fill (-2,-1) circle (3pt);
    \fill (-1,-1) circle (3pt);
    \fill (2,-1) circle (3pt);
    \fill (1,-1) circle (3pt);
    \fill (-1,-0.4) circle (3pt);
    \fill (-0.4,-1) circle (3pt);
    \fill (-0.4,-0.4) circle (3pt);
    \fill (0,0) circle (3pt);
\end{scope}
\end{tikzpicture}
\caption{Nontrivial five-twists from adding edges to the upper leftmost graph.}
\label{fig:G1a}
\end{figure}

After completion, also the first graph in the family has a four-vertex split. In Figure \ref{fig:P9_103} this split is indicated by the gray vertices.
So, up to eleven loops, all graphs with a nontrivial five-twist also have at least one nontrivial standard twist.
We hence need to prove independence of the five-twist from previously knon identities. It suffices to do this for the example shown in Figure \ref{fig:P9_103}
\begin{figure}
\begin{tikzpicture}[scale=0.7]
\begin{scope}[local bounding box=P9_103]
    \draw[li] (-4,-1.4) .. controls (-4,4) and (0.5,3.7) .. (2,2);
    \draw[white,line width=8pt] (-2,2) .. controls (-6,0) and (-6,-3) .. (-1.4,-4);
    \draw[li] (-1.4,-1.4) -- (-1.4,-4) -- (-4,-1.4) -- (-2,-2);
    \draw[white,line width=8pt] (-2,-2) -- (0,-2);
    \draw[li] (-2,-2) -- (2,-2) -- (-2,2) -- (-2,-2);
    \draw[li] (-1.4,-4) -- (2,-2);
    \draw[white,line width=8pt] (-2,-2) .. controls (1,-5.7) and (5.7,-1) .. (2,2);
    \draw[li] (-2,2) .. controls (-6,0) and (-6,-3) .. (-1.4,-4);
    \draw[li] (-2,-2) .. controls (1,-5.7) and (5.7,-1) .. (2,2);
    \draw[li] (2,-2) -- (2,2) -- (-1.4,-1.4);
    \draw[li] (-2,-0.8) -- (-0.8,-2);
    \draw[li] (-0.8,-2) -- (-0.8,-0.8) -- (-2,-0.8);
    \draw[li] (-4,-1.4) -- (-2,2);
    \fill (-2,-2) circle (3pt);
    \fill[black!30] (-2,-0.8) circle (3pt);
    \fill (-2,2) circle (3pt);
    \fill[black!30] (-0.8,-2) circle (3pt);
    \fill (-0.8,-0.8) circle (3pt);
    \fill (-1.4,-1.4) circle (3pt);
    \fill (2,-2) circle (3pt);
    \fill[black!30] (0,0) circle (3pt);
    \fill (2,2) circle (3pt);
    \fill[black!30] (-1.4,-4) circle (3pt);
    \fill(-4,-1.4) circle (3pt);
\end{scope}
    \node[below=-0.5 of P9_103] {$P_{9,103}$};
\end{tikzpicture}
\caption{The only four vertex split of the $\phi^4$ graph $P_{9,103}$ is at the gray vertices. The twist at these vertices differs from the five-twist along the
subgraph that is depicted at the left of Figure \ref{fig:G1}. The vertex $\infty$ is the center of the small square.}
\label{fig:P9_103}
\end{figure}

\begin{lem}\label{lemnewid}
The five-twist is an independent identity in $\phi^4$ theory. In general, it cannot be obtained by chains of known identities.
\end{lem}
\begin{proof}
We use the graph $P_{9,103}$ in Figure \ref{fig:P9_103}. The graph has no planar decompletion. The only four-vertex split is indicated by gray vertices in Figure \ref{fig:P9_103}.
A twist of the subgraph does not lead to further transformations, so that the classical identities give an equivalence class of two graphs. The five-twist along the left graph
in Figure \ref{fig:G1} gives a transformation that is not in this equivalence class.
\end{proof}

All calculations were done using the procedure {\tt fivetwist} in the Maple package {\tt HyperlogProcedures} \cite{Shlog}.
We generated all five-twists of $\phi^4$ periods up to loop order twelve on a small server. Because of its simple combinatorial nature, the generation of five-twists is fast.
Only the special implementation of graphs in Maple led to the fact that the program uses significant amounts of memory at high loop orders (making use of the server convenient).

\end{document}